\documentclass[sigconf]{acmart}

\usepackage{microtype}
\usepackage{graphicx}
\usepackage{booktabs} 
\usepackage{xspace}
\usepackage{xcolor}
\usepackage{amsmath}
\usepackage{amsfonts}
\usepackage{amsthm}
\usepackage{mathtools}
\usepackage{balance}
\usepackage{enumitem}
\usepackage[utf8x]{inputenc}
\usepackage{seqsplit}
\usepackage{hyperref}

\usepackage{algorithm}
\usepackage{caption}
\usepackage{subcaption}
\usepackage[noend]{algpseudocode}

\newcommand{\system}{SwiftFusion\xspace}
\newcommand{\systemShort}{SFU\xspace}
\newcommand{\tasShort}{TAS\xspace}
\newcommand{\torus}{Torus Attention\xspace}
\newcommand{\torusShort}{Torus\xspace}

\newtheorem{theorem}{Theorem}[section]

\newtheorem{lemma}[theorem]{Lemma}

\newcommand{\DefineEval}[2]{%
  \expandafter\newcommand\csname eval-#1\endcsname{#2}%
}
\newcommand{\Eval}[1]{\csname eval-#1\endcsname}

\algrenewcommand\algorithmicrequire{\textbf{Parameters:}}
\algrenewcommand\algorithmicindent{0.8em}

\makeatletter
\renewcommand{\ALG@name}{\small Algorithm}
\makeatother

\DefineEval{avg:method=sfu}{1.35$\times$}
\DefineEval{max:method=sfu}{1.77$\times$}
\DefineEval{avg:method=tas}{1.27$\times$}
\DefineEval{max:method=tas}{1.64$\times$}

\DefineEval{avg:diffdist_method=sfu}{1.61$\times$}
\DefineEval{max:diffdist_method=sfu}{3.11$\times$}
\DefineEval{avg:diffdist_method=tas}{1.47$\times$}
\DefineEval{max:diffdist_method=tas}{2.54$\times$}

\DefineEval{avg:D=32_method=tas}{1.05$\times$}
\DefineEval{max:D=32_method=tas}{1.30$\times$}
\DefineEval{avg:D=32_method=sfu}{1.12$\times$}
\DefineEval{max:D=32_method=sfu}{1.39$\times$}
\DefineEval{avg:D=64_method=tas}{1.16$\times$}
\DefineEval{max:D=64_method=tas}{1.42$\times$}
\DefineEval{avg:D=64_method=sfu}{1.28$\times$}
\DefineEval{max:D=64_method=sfu}{1.62$\times$}
\DefineEval{avg:D=128_method=tas}{1.23$\times$}
\DefineEval{max:D=128_method=tas}{1.52$\times$}
\DefineEval{avg:D=128_method=sfu}{1.32$\times$}
\DefineEval{max:D=128_method=sfu}{1.80$\times$}
\DefineEval{avg:L=96k_method=tas}{1.41$\times$}
\DefineEval{max:L=96k_method=tas}{1.52$\times$}
\DefineEval{avg:L=96k_method=sfu}{1.59$\times$}
\DefineEval{max:L=96k_method=sfu}{1.80$\times$}
\DefineEval{avg:L=128k_method=tas}{1.14$\times$}
\DefineEval{max:L=128k_method=tas}{1.20$\times$}
\DefineEval{avg:L=128k_method=sfu}{1.27$\times$}
\DefineEval{max:L=128k_method=sfu}{1.40$\times$}
\DefineEval{avg:L=160k_method=tas}{1.11$\times$}
\DefineEval{max:L=160k_method=tas}{1.25$\times$}
\DefineEval{avg:L=160k_method=sfu}{1.07$\times$}
\DefineEval{max:L=160k_method=sfu}{1.17$\times$}
\DefineEval{avg:L=192k_method=tas}{1.06$\times$}
\DefineEval{max:L=192k_method=tas}{1.20$\times$}
\DefineEval{avg:L=192k_method=sfu}{1.14$\times$}
\DefineEval{max:L=192k_method=sfu}{1.23$\times$}
\DefineEval{avg:B=1_method=tas}{1.14$\times$}
\DefineEval{max:B=1_method=tas}{1.20$\times$}
\DefineEval{avg:B=1_method=sfu}{1.27$\times$}
\DefineEval{max:B=1_method=sfu}{1.40$\times$}
\DefineEval{avg:B=2_method=tas}{1.15$\times$}
\DefineEval{max:B=2_method=tas}{1.15$\times$}
\DefineEval{avg:B=2_method=sfu}{1.27$\times$}
\DefineEval{max:B=2_method=sfu}{1.34$\times$}
\DefineEval{avg:B=4_method=tas}{1.04$\times$}
\DefineEval{max:B=4_method=tas}{1.10$\times$}
\DefineEval{avg:B=4_method=sfu}{1.13$\times$}
\DefineEval{max:B=4_method=sfu}{1.23$\times$}

\setlist{leftmargin=4mm}

\author{Jiacheng Yang}
\affiliation{%
  \institution{University of Toronto \&\\Vector Institute}
  \country{}
}
\author{Jun Wu}
\affiliation{%
  \institution{Amazon}
  \country{}
}
\author{Yaoyao Ding}
\affiliation{%
  \institution{University of Toronto \&\\
  Vector Institute \& NVIDIA}
  \country{}
}
\author{Zhiying Xu}
\affiliation{%
  \institution{Amazon}
  \country{}
}
\author{Yida Wang}
\affiliation{%
  \institution{Amazon}
  \country{}
}
\author{Gennady Pekhimenko}
\affiliation{%
  \institution{University of Toronto \&\\
  Vector Institute \& NVIDIA}
  \country{}
}

\acmYear{2026}\copyrightyear{2026}
\acmConference[ACM CAIS '26]{ACM Conference on AI and Agentic Systems}{May 26--29, 2026}{San Jose, CA, USA}
\acmBooktitle{ACM Conference on AI and Agentic Systems (ACM CAIS '26), May 26--29, 2026, San Jose, CA, USA}
\acmDOI{10.1145/3786335.3813174}
\acmISBN{979-8-4007-2415-2/26/05}

\ccsdesc[500]{Computing methodologies~Distributed algorithms}
\ccsdesc[500]{Computing methodologies~Computer vision}
\ccsdesc[300]{Computing methodologies~Machine learning algorithms}
\ccsdesc[300]{Computing methodologies~Parallel algorithms}

\keywords{Sequence Parallelism, Diffusion Transformers, GPUs}

\begin{document}

\title{\system: Scalable Sequence Parallelism for Distributed Inference of Diffusion Transformers on GPUs}

\begin{abstract}
Diffusion Transformers (DiTs) have gained increasing adoption in high-quality image and video generation. 
As demand for higher-resolution images and longer videos increases, single-GPU inference becomes inefficient due to increased latency and large activation sizes. 
Current frameworks employ sequence parallelism (SP) techniques such as Ulysses Attention and Ring Attention to scale inference.
However, these implementations have three primary limitations: (1) suboptimal communication patterns for network topologies on modern GPU machines, (2) latency bottlenecks from all-to-all operations in inter-machine communication, and (3) GPU sender-receiver synchronization and computation overheads from using two-sided communication libraries.
To address these issues, we present \system, a topology-aware efficient DiT serving engine. \system incorporates three key innovations: (1) a topology-aware sequence parallelism technique that accounts for inter- and intra-machine bandwidth differences, (2) Torus Attention, a novel SP technique enabling overlapping of inter-machine all-to-all operations with computation, and (3) a one-sided communication implementation that minimizes GPU sender-receiver synchronization and computation overheads. Our experiments demonstrate that \system outperforms the state-of-the-art approach by an average of \Eval{avg:method=sfu} (up to \Eval{max:method=sfu}).
\end{abstract}
\maketitle

\section{Introduction}

Diffusion transformers (DiTs) have been widely applied to generate photo-realistic images and videos, which drives the demand for computation and memory.
This demand has been amplified by the recent advances in generating high-resolution images \cite{flux-image-gen} or long-duration videos \cite{cogvideox,opensora2}, causing a long sequence length in the attention layer.
Thus, current DiT inference engines use sequence parallelism (SP) to shard the sequence dimension and then execute it on multiple GPUs.

However, existing SP techniques suffer from high communication overheads, since their design either does not fully exploit the network topology, especially the bandwidth difference between the intra- and inter-machine network, or causes bubbles in computation pipelines due to non-overlapping communication overheads.
Specifically, in Ring Attention \cite{ring-attn} the transfer of key and value tensors from/to its neighboring GPUs can be overlapped with the attention computation.
However, the communication volume per GPU required by Ring Attention does not decrease with increasing number of GPUs, thus making DiT inference communication-bound when scaling up to multiple GPUs.
Ulysses Attention \cite{ulysses} reduces the communication volume with more GPUs by employing all-to-all operators to first gather the full sequence dimension on each device.
However, the all-to-all operators are not overlapped by computation operators, and the parallelization degree of Ulysses Attention is limited by the number of heads of the attention layer.
Unified sequence parallelism (USP) \cite{usp} combines Ring and Ulysses Attention by applying Ring Attention for inter-machine communication and Ulysses Attention for intra-machine communication, thus achieving the best of both methods.
This design is intuitive as the intra-machine networks on modern GPU machines are often fully connected (e.g. NVSwitch \cite{nvswitch}) and thus enable simultaneous and high-bandwidth data exchange which is required by the all-to-all operators in Ulysses Attention.
However, they do not optimize the high inter-machine communication overhead introduced by Ring Attention, hindering scalability to multiple GPU machines.
These issues make the distributed DiT inference communication-bound, limiting its scalability across multiple GPU machines.

In this paper, we propose \system to reduce communication overheads when scaling up distributed DiT inference to multiple GPU machines, which is built upon three key ideas.
First, we introduce \emph{topology-aware communication scheduling}.
Unlike USP, we apply Ulysses Attention for inter-machine communication and Ring Attention for intra-machine communication.
In this design, we utilize the high-bandwidth intra-machine network for conducting the high-volume communication incurred by Ring Attention at the cost of using inter-machine networks for Ulysses Attention.
Second, to further reduce the high inter-machine communication costs introduced by the all-to-all operators in Ulysses Attention, we observe there are stationary elements before and after the all-to-all operators, which enables us to overlap the all-to-all operators with the attention computation on those stationary elements.
Based on this observation, we propose \emph{\torus}, a novel distributed attention algorithm that breaks down the attention computation and all-to-all operators into multiple chunks and overlaps the computation of the current chunk with communication of the next chunk.
Third, we leverage \emph{one-sided communication} libraries such as NVSHMEM \cite{nvshmem} and develop a specialized communication algorithm that unifies \torus, Ulysses, and Ring Attention to further reduce the synchronization and computation overheads caused by traditional two-sided communication libraries.

We evaluate \system on commonly-used image and video workloads and our results show that \system reduces inter-machine communication overheads and improves end-to-end inference latency across a range of distributed configurations (see Section \ref{sec:overall-performance-exp}) by on average \Eval{avg:method=sfu} speedup (up to \Eval{max:method=sfu}) over USP.

Our contributions can be summarized as follows:
\begin{itemize}[itemsep=0.5em,parsep=0em, topsep=0em]
    \item We present topology-aware communication scheduling that better exploits the network topologies on modern GPU machines.
    \item We design \torus, a novel distributed attention algorithm that hides inter-machine communication latency for attention computation.
    \item We introduce a one-sided communication implementation that unifies \torusShort, Ulysses, and Ring Attention to minimize synchronization and computation overheads.
    \item We evaluate \system on image and video generation workloads and achieve a state-of-the-art speedup in end-to-end inference latency compared to prior works.
\end{itemize}

\section{Background}
\label{sec:background}

\subsection{Diffusion Transformers}

As shown in Figure \ref{fig:dit}, to generate an image or video, diffusion transformers (DiTs) start from a random noise tensor through multiple diffusion sampling steps.
Each sampling step evaluates the DiT to gradually denoise the noise tensor. The output tensor of the last sampling step will be fed to a variational autoencoder to produce the final generated images or videos.

DiTs typically have a small model size but a much larger activation tensor size.
For example, CogVideoX \cite{cogvideox} takes only 12 GiB to store the model weights, but causes out-of-memory errors when generating a 10s 768$\times$1360 video on a single A100 GPU (40 GiB).
This situation could be even worse when generating high-resolution images or long-duration videos, because the activation tensor size grows linearly with the sequence length.
As a result, efficient distributed inference of DiTs relies heavily on sequence parallelism techniques to distribute the activations across multiple GPUs.

\begin{figure}[t]
    \centering
    \includegraphics[width=.75\linewidth]{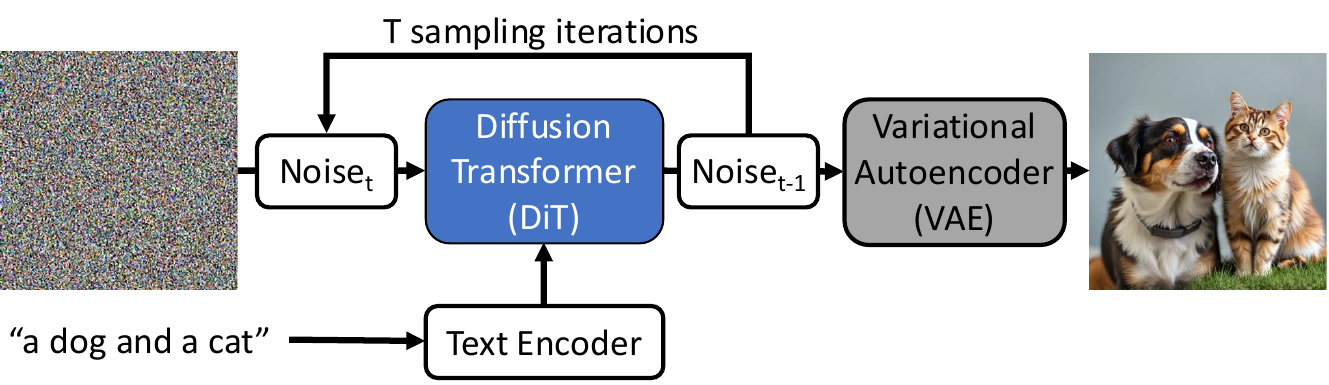}
    \caption{Generating images with diffusion transformers}
    \label{fig:dit}
\end{figure}

\subsection{Sequence Parallelisms}

\begin{figure*}
    \centering
    \includegraphics[width=\linewidth]{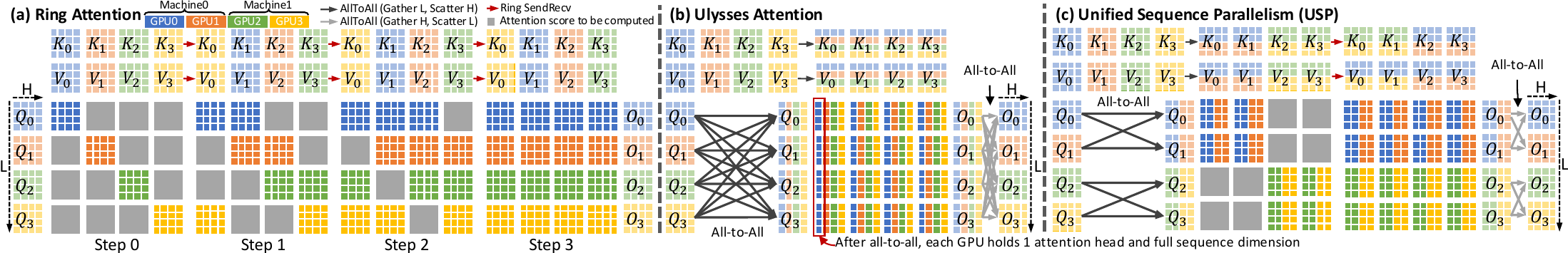}
    \caption{Illustration of commonly used SP techniques for DiT inference.}
    \label{fig:background}
\end{figure*}

Sequence parallel (SP) techniques divide and distribute the $Q$, $K$, and $V$ tensors of an attention layer to multiple GPUs.
Specifically, we represent the query, key, and value tensors of an attention layer as $Q$, $K$, and $V$ respectively, each of which has a shape $[B, L, H, D]$, where $B$ stands for batch size, $L$ for the sequence length, $H$ for the number of heads, and $D$ for the hidden dimension per attention head.
With $P$ GPUs, SP techniques shard the $Q$, $K$, $V$ tensors along the sequence dimension, resulting in tensors with shape $[B, L/P, H, D]$ on each GPU.
To avoid ambiguity, we assume the communication volume is calculated in terms of communication \emph{per GPU}.

As the $Q$ tensor on each GPU needs to compute the attention scores across all $KV$ tensors, SP techniques require necessary communication to compute the attention scores between each $Q$ tensor and all the $KV$ tensors.
Prior works achieve this by Ring Attention \cite{ring-attn}, Ulysses Attention \cite{ulysses}, or Unified Sequence Parallelism (USP) \cite{usp}, the combination of these two approaches.

\textbf{Ring Attention}~ Figure \ref{fig:background}a shows the execution of Ring Attention \cite{ring-attn}, which organizes all GPU communication in a ring-like manner.
Specifically, for $P$ GPUs, Ring attention consists of $P$ steps. 
In every step, GPU $i$ sends its own local $KV$ tensors to neighboring GPU $(i + 1)\% P$ and receives tensors from GPU $(i - 1) \% P$. 
At the same time, GPU $i$ computes attention scores between its local $Q$ tensor and the received $KV$ tensors.
To aggregate the results of the local $Q$ tensor on different received $KV$ tensors, Ring Attention maintains the running maximum and sum tensors of the attention outputs on each GPU similar to FlashAttention \cite{flash-attn,flash-attn-2,flash-attn-3}.
After $P$ steps, each GPU has seen all the $KV$ tensors from all the other GPUs, and thus can compute the final output tensor correctly.
In this way, Ring Attention needs to send or receive $\frac{2BLHD}{P}$ elements to transfer the $KV$ tensors to its neighbors for all but the last step, leading to $\frac{2(P - 1)BLHD}{P} \approx{2BLHD}$ elements in total.

\textbf{Ulysses Attention}~ Figure \ref{fig:background}b shows the execution of Ulysses Attention.
The key idea behind Ulysses Attention is that the attention computation is head-independent, i.e., the outputs of an attention layer of one head do not depend on the other heads of the $QKV$ tensors.
Thus, Ulysses Attention can first apply three all-to-all operators on the $Q$, $K$, and $V$ tensors, to gather the full sequence dimension and scatter the head dimension on each GPU, leading to $Q$, $K$, and $V$ tensors on each GPU with a shape $[B, L, H/P, D]$ respectively.
By doing so, the attention computation can be done locally on each GPU without further communication.
Note that after the attention computation, the output tensor $O$ on each GPU needs to be recovered to its original form, i.e., having a shape $[B, L/P, H, D]$.
Thus, Ulysses Attention executes another all-to-all operator by gathering the head dimension and scattering the number of head dimension.
Since there are four all-to-all operators, i.e. on the $Q$, $K$, $V$, and $O$ tensors, we can deduce that Ulysses Attention requires $\frac{4(P - 1)BLHD}{P^2} \approx \frac{4BLHD}{P}$ elements for communication.
When $P = 2$, Ulysses Attention requires $BLHD$ elements to communicate, which the same as that of Ring Attention.
Otherwise, Ulysses Attention requires less communication volume than Ring Attention.
However, Ulysses Attention is not applicable when the number of heads $H$ is smaller than the number of GPUs $P$, limiting its scalability to a large number of GPUs \cite{usp}.

\textbf{Unified Sequence Parallelism}~ To scale to a larger number of GPUs without incurring high communication overhead, Unified Sequence Parallelism (USP) \cite{usp} proposes to combine Ulysses and Ring Attention. 
Figure \ref{fig:background}c shows USP execution with two GPU machines, each of which is equipped with two GPUs.
USP applies Ulysses Attention for the intra-machine communication and Ring Attention for inter-machine communication. 
Specifically, USP first applies all-to-all operators within each GPU machine to gather the sequence dimension and scatter the head dimension.
However, unlike Ulysses Attention which scatters the whole head dimension and gathers the full sequence length, USP only scatters half head dimension within each GPU machine and only gathers half of the sequence dimension.
To make sure the $Q$ tensor on each GPU has seen the other half of the sequence dimension of the $KV$ tensors, USP applies Ring Attention to send and receive the $KV$ tensors on its neighbor \emph{GPU machine}.
Finally, USP applies another all-to-all operator to gather the head dimension and scatter the sequence dimension which restores the output tensors to its original distributed form.
In this way, the all-to-all operators only run within the GPU machine which is accelerated by high-bandwidth intra-machine networks such as NVSwitch \cite{nvswitch}, and the communication in the Ring Attention overlaps with the computation.

\subsection{One-Sided Communication Libraries}

The cross-GPU communication required by SP techniques typically relies on P2P communication functions in two-sided GPU communication libraries such as NCCL \cite{nccl}.
These libraries automatically detect and adapt communication patterns to the network topology and simplify the programmer's efforts to implement GPU communications.
However, using these libraries incurs implicit synchronization barriers as the receiver needs to ensure the sender's data is ready and vice versa.
Recently, one-sided communication libraries such as NVSHMEM \cite{nvshmem} have achieved decent speedups in the inference of large models \cite{deepep,flux-comm-overlap}.
Instead of using send-receive primitives, these libraries rely on push and pull primitives for transferring data from one GPU to another.
These libraries remove per-transfer synchronization between sender and receiver, placing synchronization responsibility on the programmer and exposing more flexibility for communication optimization.

\section{Key Challenges}
\label{sec:motivation}

We carefully examine existing SP techniques and identify three key challenges for high-performance DiT inference.

\begin{figure}[t]
    \centering
    \includegraphics[width=.9\linewidth]{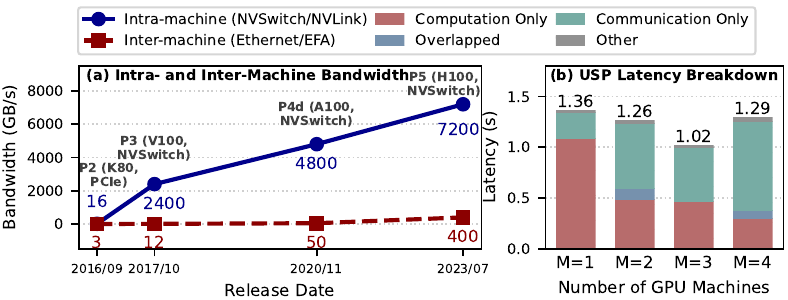}
    \caption{(a) Differences in aggregated bandwidth between intra- and inter-machine network; and (b) Latency breakdown of the state-of-the-art SP technique USP \cite{usp}.}
    \label{fig:bandwidth}
\end{figure}

\begin{figure}[t]
    \centering
    \includegraphics[width=.8\linewidth]{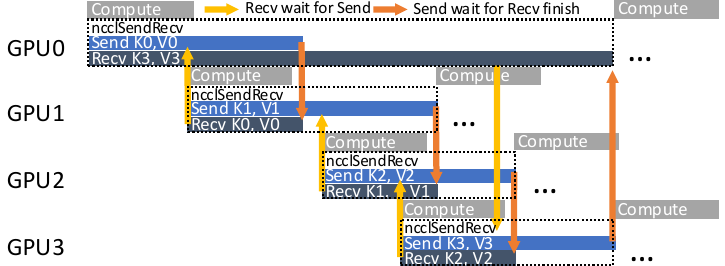}
    \caption{Implicit synchronization in each step of the Ring Attention due to two-sided communication APIs, i.e, \texttt{ncclSendRecv}.}
    \label{fig:ring-sync}
\end{figure}
\noindent\textbf{Challenge 1: Inter- and Intra-Machine Bandwidth Difference}~
Figure \ref{fig:bandwidth}a shows the difference between inter- and intra-machine communication bandwidth on modern GPU machines with respect to their release dates. 
Recently, high performance GPU-aware inter-machine communication networks have advanced considerably. Techniques such as GPUDirect Remote Direct Memory Access (RDMA) on Infiniband network and AWS Elastic Fabric Adapter (EFA) can have bandwidth ranging from 400 to 3200 Gbps. However, intra-machine communications are still much more performant than inter-machine communications in terms of both latency and throughput.
Existing SP implementations either do not optimize for network bandwidth differences or adopt a suboptimal combination of the two SP algorithms, leading to high communication overheads and low DiT inference system performance.
For example, as shown in Figure \ref{fig:bandwidth}b, the state-of-the-art SP technique, USP \cite{usp}, becomes communication-bound when scaling up to four GPU machines. 
This is because USP uses Ring Attention across GPU machines and Ulysses Attention within each GPU machine.
This design maps naturally to the network topology as Ring Attention only introduces ring-like communication patterns that are well-suited for inter-machine network and the more complicated Ulysses Attention is handled by the faster intra-machine GPU interconnections.
However, as discussed in Section \ref{sec:background}, the communication volume of Ring Attention does not decrease with more GPU machines, which produces high inter-machine communication overheads and limits its scalability to multiple GPU machines.
To this end, we propose to leverage Ulysses Attention for inter-machine communication and Ring Attention for intra-machine communication to better exploit the network topology on modern GPU machines.

\noindent\textbf{Challenge 2: Non-Overlapping All-to-All Operators in Ulysses Attention}~
As discussed in Section \ref{sec:background}, Ulysses Attention leverages all-to-all operators to gather the full sequence length on each GPU before performing attention computation.
While the latency of all-to-all operators in the Ulysses Attention reduces almost linearly with the number of GPUs, Ulysses Attention does not support overlapping its all-to-all operators with computation.
This is because existing implementations treat all-to-all operators as single atomic communication operations.
As a result, SP techniques that rely on Ulysses Attention can still be bottlenecked by the communication overheads, especially with a large number of GPU machines.
This issue becomes even more dominant when choosing Ulysses Attention for inter-machine communication.
To address this challenge, we propose chunked all-to-all operators where we pipeline the communication of each chunk with the computation of attention on the already received chunks.

\noindent\textbf{Challenge 3: Overheads from Two-Sided Communication Libraries}~
Existing SP implementations, such as Ring Attention and Ulysses Attention, rely on two-sided communication libraries such as NCCL \cite{nccl} for inter-GPU communication.
However, these libraries require strict synchronization between the sender and receiver GPUs to ensure the liveness of the send and receive data buffers, which leads to significant synchronization overhead.
For example, as shown in Figure \ref{fig:ring-sync}, in each step of ring attention, each GPU needs to wait for the data from the previous GPU before it can proceed with the attention computation and send data to the next GPU, which implicitly synchronizes all GPUs at each step.
Furthermore, two-sided communication libraries for intra-machine communication are often implemented with CUDA kernels.
These kernel-based implementations do not utilize driver-level communication APIs (e.g., \texttt{cudaMemcpyPeer}) and consume GPU Streaming Multiprocessors (SMs), causing contention with the model's compute kernels and increasing computation overheads.
As a result, the overall performance is degraded due to the synchronization and computation overhead, especially when the number of GPUs increases.
To address this challenge, we leverage one-sided communication libraries such as NVSHMEM \cite{nvshmem} and carefully design communication patterns to minimize synchronization and computation overhead while ensuring data consistency and correctness.

\section{\system}

\subsection{Overview}
\label{sec:key-ideas}

We propose \system which achieves efficient DiT inference via the following three key ideas.

\noindent\textbf{Topology-Aware Sequence Parallelism}~
To address the inter- and intra-machine bandwidth difference challenge, \system proposes to combine Ulysses and Ring Attention in a topology-aware manner.
Specifically, \system leverages Ulysses Attention for inter-machine communication and Ring Attention for intra-machine communication, since Ulysses Attention introduces lower communication volume than Ring Attention when scaling to multiple GPU machines and is more suitable for the slower inter-machine networks.
On the other hand, Ring Attention introduces higher communication volume than Ulysses Attention when scaling to multiple GPUs.
Thus, Ring Attention is better suited for the faster intra-machine GPU interconnections.
In this way, \system effectively reduces the overall communication overheads and improves system performance for distributed DiT inference.
However, since \system requires Ulysses Attention which has non-overlapping all-to-all operators and uses two-sided communication libraries that introduce high synchronization and computation overheads, \system further addresses these two challenges with the following ideas. 

\noindent\textbf{Communication-Overlapped \torus}~
To address the non-overlapping inter-machine communication overheads challenge in Ulysses Attention, we propose \torus that divides the all-to-all operators into multiple chunks and pipelines the communication of each chunk with the computation of attention on the already received chunks.
Specifically, we break the $QKV$ tensors into multiple chunks.
When performing the all-to-all operators in Ulysses Attention, we observe that the chunks for the current rank remain unchanged after the all-to-all operators.
Thus, \system can directly first start the attention computation on those chunks without waiting for the all-to-all operators to finish and at the same time starts sending and receiving the other chunks to and from other GPUs.
While receiving each chunk, \system immediately starts the attention computation on the already received chunks, effectively overlapping the communication and computation.
By doing so, \system effectively hides the inter-machine communication latency and reduces the overall communication overheads.

\noindent\textbf{One-Sided Communication Pattern Design}~
To effectively leverage one-sided communication libraries, \system tailors the communication patterns of Ring, Ulysses, and \torus to fully utilize the benefits of one-sided communication while minimizing synchronization overheads.
Specifically, \system carefully co-designs the communication algorithms and the direction of each communication operator.
As a result, \system requires only intra-machine synchronizations and two inter-machine synchronizations (at the beginning and end of a layer), thus efficiently utilizing one-sided libraries and minimizing synchronization and computation overheads.

\subsection{Topology-Aware Communication Scheduling}

\begin{figure}[t]
    \centering
    \includegraphics[width=.8\linewidth]{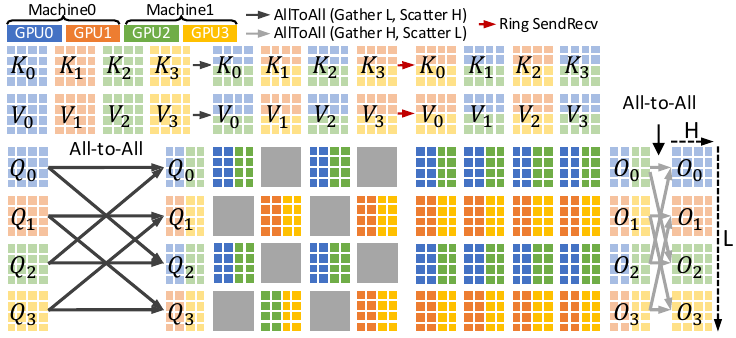}
    \caption{Topology-aware sequence parallelism in \system.}
    \label{fig:tas}
\end{figure}
\system employs a topology-aware communication scheduling strategy that assigns SP techniques based on the network topology of the GPU cluster.
Unlike USP \cite{usp}, \system applies Ulysses Attention for inter-machine communication and Ring Attention for intra-machine communication.
Specifically, in Figure \ref{fig:tas}, suppose there are $N$ GPU machines, each of which is equipped with $M$ GPUs, and \system organizes the $N \times M$ GPUs into a 2D device mesh of shape $P_u \times P_r$, creating a process group of size $P_u$ for Ulysses Attention and a process group of size $P_r$ for Ring Attention.
As discussed in Section \ref{sec:background}, $P_u$ is limited by the number of heads $H$ as Ulysses Attention requires $H$ to be divisible by $P_u$.
In the simplest case $H = N$, we can set $P_u = N$ and $P_r = M$, where \system applies Ulysses Attention across the $N$ GPU machines and Ring Attention within each GPU machine.
In this way, \system minimizes the inter-machine communication volume and effectively utilizes the high-bandwidth intra-machine GPU interconnections.
However, when $H \neq N$, we set $P_u = \gcd(NM, H)$\footnote{$\gcd$ stands for the greatest common divisor.} as this maximizes the usage of Ulysses Attention to reduce inter-machine communication volume, and then set $P_r = \frac{NM}{P_u}$.
By doing so, \system can always leverage Ulysses Attention to reduce inter-machine communication volume as much as possible:
(1) When $P_u \geq N$, \system fully utilizes Ulysses Attention for inter-machine communication, and the communication volume decreases almost linearly with the number of machines $N$;
(2) When $P_u < N$, \system uses a combination of Ulysses and Ring Attention for inter-machine communication volume, and the Ulysses Attention can still be used to reduce inter-machine communication volume.
In contrast, even with the same $P_u$ and $P_r$, USP uses Ring Attention for inter-machine communication, which implies that Ulysses Attention reduces the inter-machine communication volume only if $P_r < N$.
The only case where \system could incur more inter-machine communication than USP is when $P_u = 2$, where Ulysses has more communication volume than Ring Attention as described in Section \ref{sec:background}.
Thus, USP incurs higher inter-machine communication volume than \system in most cases and additionally we formally prove it in Appendix \ref{app:complexity-analysis}.

\begin{figure*}[t]
    \centering
    \includegraphics[width=\linewidth]{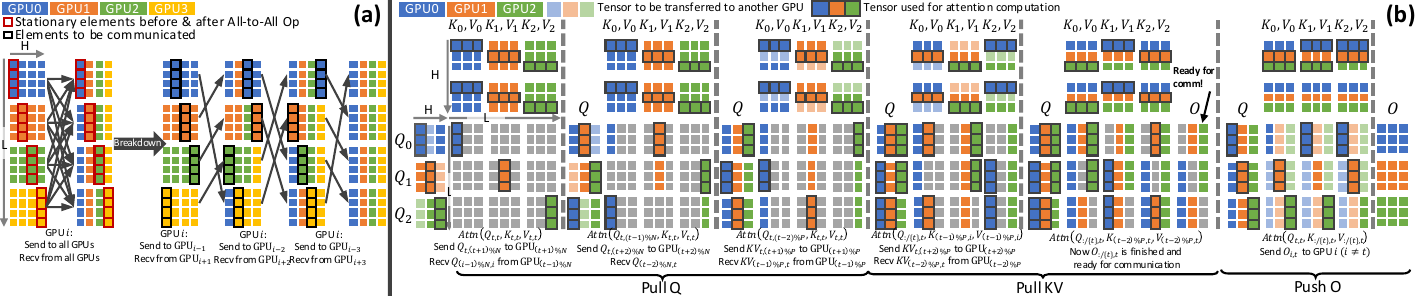}
    \caption{Illustration of (a) breakdown of an all-to-all operator and (b) \torus communication scheduling with $3$ GPU machines.}
    \label{fig:torus-scheduling}
\end{figure*}

During the attention execution, since each $QKV$ tensor is sharded along the sequence dimension, \system first applies all-to-all operators to gather the sequence dimension and scatter the head dimension of $QKV$ tensors within each GPU machine with the Ulysses Attention process group.
Then, \system applies Ring Attention to communicate the $KV$ tensors within the GPUs on each GPU machine.
Finally, \system applies another all-to-all operator within each GPU machine to gather the head dimension and scatter the sequence dimension of the output tensor $O$ to restore the output tensor to its original distributed form.
While the exact speedup of \system depends on the deployment environment and the available intra- and inter-machine bandwidths, we expect \tasShort can significantly reduce inter-machine communication overhead and thereby improve the performance of distributed DiT inference, especially when the discrepancy between intra- and inter-machine bandwidth is huge.

However, since Ulysses Attention relies on all-to-all operators that are not overlapped with computation, using Ulysses Attention for inter-machine communication incurs non-overlapping inter-machine communication overheads (see Section \ref{sec:motivation}).
We address this challenge with the following \torus.

\subsection{\torus}
To reduce the non-overlapping inter-machine communication overheads, we propose \torus that hides the latency of all-to-all operators by overlapping their communication with attention computation.

\textbf{Breakdown of All-to-All Operators}~
We find that some elements of the input and output tensors are stationary before and after the all-to-all operators, so they can be used for attention computation without waiting for the all-to-all to finish. At the same time, we start the remaining all-to-all transfers in the background to overlap the communication with computation.

Specifically, Figure \ref{fig:torus-scheduling}a (left) illustrates the distributed layout of a tensor before and after an all-to-all operator among four GPUs with sequence length $L = 16$ and head dimension $H = 4$.
We found that elements with red boxes, namely elements whose index of the number of heads dimension is the current rank, will remain stationary before and after the all-to-all operator in the Ulysses Attention.
Thus, these elements can be directly used for some computation without waiting for the all-to-all operators to finish.
At the same time, we can start the all-to-all operators for the remaining elements in the background.
Thus, we can hide the latency of the all-to-all operator by overlapping with computation.
Based on this observation, we propose \torus that breaks down a single all-to-all operator into multiple stages.
Specifically, as shown in Figure \ref{fig:torus-scheduling}a (right), the all-to-all operators can be broken into 2 stages for 3 GPU machines.
In each stage $k \in \{1, 2\}$, GPU $i$ first uses elements whose head index equals $i$ for computation, and concurrently sends elements with head index $(i - k) \% N$ to GPU $(i - k) \% N$.
In this way, after the communication of each stage, we initiate computation on the newly received elements while sending the elements required by the computation of the next stage in the background, thus effectively overlapping the all-to-all operators with computation.

\textbf{Communication Scheduling for Attention Computation}~
While breaking all-to-all operators hides the latency of \emph{one} all-to-all operator with computation, hiding latency of \emph{four} all-to-all operators in Ulysses Attention remains challenging. 
Besides, the first three all-to-all operators transform the $Q$, $K$, and $V$ tensors respectively, while the fourth all-to-all operator transforms the output tensor $O$ (see Section \ref{sec:background}).
These tensors have different lifespans during attention computation.
Thus, we need to carefully schedule the communication and computation order of each tensor chunks to overlap communication with computation as much as possible.

To address this challenge, we design a tailored communication scheduling for overlapping the all-to-all operators with attention computation.
The key idea is to first schedule the communication of the $Q$ tensors first followed by the $KV$ tensors, since the communication of $KV$ tensors doubles the communication volume and makes it more difficult to be fully overlapped with computation.
Specifically, as shown in Figure \ref{fig:torus-scheduling}b, during the attention execution, we break the Ulysses Attention execution into the \emph{Pull Q}, \emph{Pull KV}, and \emph{Push O} stages.
The name of ``\emph{Pull}'' and ``\emph{Push}'' indicate the communication primitives we use in the one-sided communication implementation as described in Section \ref{sec:1-side-comm}.
However, this algorithm can also be implemented with two-sided communication libraries such as NCCL \cite{nccl}.
For clarity, we denote $T_{l, h}$ as the partition of the sequence dimension with index $l$ and the partition of the number of head dimension with index $h$ of the tensor $T$ respectively.
We use ``$_:$'' to denote all indices and ``$_{:\setminus\{l\}}$'' to denote the set of all indices except $l$. Thus $T_{:\setminus\{l\}, h}$ denotes the slice of $T$ that contains all sequence partitions except the $l$-th, while selecting the $h$-th partition of the head dimension.
Unless otherwise noted, we use GPU $t$ to denote the GPU with local rank $t$ in the \torus process group.

In the \emph{Pull Q} stages, we start transferring the chunks of the $Q$ tensor by breaking down the all-to-all operators described above, whose latency is hidden by overlapping with the attention computation.
Specifically, we have exactly $N$ Pull Q stages.
In the $k$-th Pull Q stage ($k \in \{1, \dots N\}$), GPU $t$ first computes the attention scores using the $Q_{(t - k + 1)\% N,t}$, $K_{t, t}$, and $V_{t, t}$.
At the same time, the GPU $t$ will send $Q_{t, (t + k) \% N}$ to the GPU $(t + k) \% N$, where the communication of the $Q$ tensor is overlapped with computation.
Note that in the last Pull Q stage, GPU $t$ computes the attention using the received $Q_{(t - N + 1)\%N,t}$ while sending $K_{t, (t + 1) \% N}$ and $V_{t, (t + 1) \% N}$ to GPU $(t + 1) \% N$ for preparing the computation in the first Pull KV stage.

In the \emph{Pull KV} stages, similar to the Pull Q stages, we break down the all-to-all operators of the $KV$ tensors into multiple stages and hide the latency of each stage communication with computation.
Specifically, we have $N - 1$ Pull KV stages.
In the $k$-th Pull KV stage ($k \in \{1, \dots N - 2\}$), GPU $t$ first computes the attention using the received $Q_{:\setminus\{t\},t}$, $K_{(t - k)\% N,t}$, and $V_{(t - k)\% N,t}$, while sending $K_{t, (t + k) \% N}$ and $V_{t, (t + k) \% N}$ to GPU $(t + k) \% N$.
In the last Pull KV stage, GPU $t$ only computes the attention on $Q_{:\setminus\{t\},t}$, $K_{(t - N + 1)\% N,t}$, and $V_{(t - N + 1)\% N,t}$ without overlapping with any communication.
Note that we need to combine the newly computed attention outputs with previously computed attention outputs using different $KV$ tensors and thus we need to merge the attention outputs similar to the Ring Attention.
During the computation, the GPU $t$ will send $K_{t, (t + k) \% N}$ and $V_{t, (t + k) \% N}$ such that the communication of the $KV$ tensors can be overlapped with computation.
After the Pull KV stages, the attention output $O_{:\setminus\{t\},t}$ on GPU $t$ is fully computed and ready to be sent to the other GPUs.

In the \emph{Push O} stages, we overlap the computation of the remaining output tensor, i.e., $O_{t, t}$.
Specifically, while computing $O_{t, t}$ locally on GPU $t$ using $Q_{t, t}$, $K_{:\setminus\{t\}, t}$, and $V_{:\setminus\{t\}, t}$, we start sending $O_{i, t}$ to GPU $i$ for all $i \neq t$.
Note that $O_{t, t}$ should stay on GPU $t$ and does not need to be transferred to the other GPUs.
Thus, we can overlap the computation of $O_{t, t}$ with the communication of $O_{:\setminus\{t\},t}$, effectively hiding the communication overheads of the last all-to-all operator.

In \system, we use \torus only for inter-machine communication.
This is because intra-machine communication overheads are already negligible due to high-bandwidth intra-machine interconnects such as NVSwitch.
For simplicity, we assume $N \mid P_u$, i.e., the number of GPU machines divides the parallelization degree of Ulysses Attention, such that we can assume the parallelization degree of \torus is $N$.
However, \system can be easily extended when $N \nmid P_u$ by only applying \torus on a subset of the GPU machines.
To avoid confusion, we denote $P'_u = P_u / N$ as the parallelization degree of Ulysses Attention for intra-machine communication and we assume $P'_u \cdot P_r = M$ where $M$ is the number of GPUs on each GPU machine.
In this way, we overlap the inter-machine communication while still leveraging the benefits of Ulysses and Ring Attention, thus significantly reducing the inter-machine communication overheads. 

\subsection{One-Sided Communication Implementation}
\label{sec:1-side-comm}

As shown in Algorithm \ref{alg:1-side-comm}, we propose a unified implementation for \torusShort, Ulysses, and Ring Attention using one-sided communication libraries such as NVSHMEM \cite{nvshmem}.
For clarity, we denote a GPU $x$ by $x = (t, u, r)$, where $t$, $u$, and $r$ are the local ranks in the Torus, Ulysses, and Ring process groups, respectively.
We use $\textsc{ScatterPush}(A, B, x)$ ($\textsc{GatherPull}(A, B, x)$) to denote element-wise push (pull) operations between tensors in list $A$ and the corresponding GPUs in list $x$.
Namely, for the $i$-th operation, we push $A[i]$ to GPU $B$ on $x[i]$ (pull $B$ from GPU $x[i]$ to $A[i]$).
The notation $(t, :, r)$ means all GPUs with local rank $t$ in the Torus process group and local rank $r$ in the Ring process group, where ``$:$'' indicates all ranks in the Ulysses group.

\begin{algorithm}[t]\small
    \caption{\small Pseudocode of \system with one-sided communication}
    \label{alg:1-side-comm}
    \begin{algorithmic}[1]
        \Require The input $QKV$ tensors of shape $[B, L, H, D]$ on each GPU; \torusShort, Ulysses, and Ring Attention degree and rank denoted by $(T, t)$, $(U, u)$, $(R, r)$
        \Procedure{RingAttn}{$Q, K, V, O, l, m$}\label{line:ringAttnBegin}
            \State Create empty tensors $K^\text{buf}, V^\text{buf}$, and tensors $K^\text{cur}, V^\text{cur}$ cloned from $K, V$
            \For{$i \gets 1, \dots, R$}
                
                \State $E_\text{i} \gets \textsc{Pull}(\{K, V\}, \{K^\text{buf}, V^\text{buf}\}, (t, u, (r + i) \% R))$~\textbf{if $i < R$}
                \State $\textsc{FlashAttention}(Q, K^\text{cur}, V^\text{cur}, O, l, m)$\Comment{In-place update $O, l, m$}
                \State $\textsc{Wait}(E_\text{i})$~\textbf{if $i < R$}
                \State $K^\text{cur}, V^\text{cur} \gets K^\text{buf}, V^\text{buf}$
            \EndFor
        \EndProcedure{}\label{line:ringAttnEnd}
        \Function{\system}{$Q, K, V, T, t, U, u, R, r$}
            \State Rearrange $Q$, $K$, $V$ from $[B, L, TU\cdot \frac{H}{TU}, D]$ into $\left[T, U, B, \frac{H}{TU}, L, D\right]$
            \State $Q^\text{buf}, K^\text{buf}, V^\text{buf} \gets \textsc{Clone}(Q, K, V)$
            \State $O^\text{buf} \gets \textsc{ZeroLike}(Q)$
            \State $O \gets \textsc{ZeroLike}(Q)$
            \State $l \gets \textsc{Zeros}(\left[T, U, B, \frac{H}{TU}, L\right])$\Comment{Running sum tensor}
            \State $m \gets \textsc{Full}(-\infty, \left[T, U, B, \frac{H}{TU}, L\right])$\Comment{Running max tensor}
            \State $\textsc{ScatterPush}(\{Q_{t, :}, K_{t, :}, V_{t, :}\}, \{Q^\text{buf}_{t, u}, K^\text{buf}_{t, u}, V^\text{buf}_{t, u}\}, (t, :, r))$\label{line:ulyssesBegin}\label{line:commOuter1}
            \State {$\textsc{BarrierAll}()$}\label{line:barrierBegin}
            \State $Q_{t, :}, K_{t, :}, V_{t, :} \gets Q^\text{buf}_{t, :}, K^\text{buf}_{t, :}, V^\text{buf}_{t, :}$
            \For{$k \gets 1, \dots, N - 1$}\label{line:issueCommBegin}
                \State $t' \gets (t-k)\% N$
                \State $E^Q_k \gets \textsc{GatherPull}(\{Q_{t', :}\}, \{Q^\text{buf}_{t', u}\}, (t', :, r))$\label{line:commOuter2}
                \State $E^{KV}_k \gets \textsc{GatherPull}(\{K_{t', :}, V_{t', :}\}, \{K^\text{buf}_{t', u}, V^\text{buf}_{t', u}\}, (t', :, r))$\label{line:commOuter3}
            \EndFor\label{line:issueCommEnd}
            \State $\textsc{RingAttn}(Q_{t, :}, K_{t, :}, V_{t, :}, O_{t, :}, l_{t, :}, m_{t, :})$\Comment{First \emph{Pull Q} stage}\label{line:firstPullQ}
            \For{$k \gets 1, \dots, N - 1$}\label{line:pullQBegin}
            \Comment{Remaining $N - 1$ \emph{Pull Q} stages}
                \State $t' \gets (t-k)\% N$
                \State $\textsc{Wait}(E^Q_k)$
                \State $\textsc{RingAttn}(Q_{t', :}, K_{t, :}, V_{t, :}, O_{t', :}, l_{t', :}, m_{t', :})$
            \EndFor\label{line:pullQEnd}
            \For{$k \gets 1, \dots, N - 1$}\Comment{\emph{Pull KV}}\label{line:pullKVBegin}
                \State $t' \gets (t-k)\% N$
                \State $\textsc{Wait}(E^{KV}_k)$~\textbf{and}~$\textsc{Barrier}(R)$\label{line:pullKVBarrier}\Comment{Barrier the Ring Attention group}
                \State $\textsc{RingAttn}(Q_{:\setminus\{t\}, :}, K_{t', :}, V_{t', :}, O_{:\setminus\{t\}, :}, l_{:\setminus\{t\}, :}, m_{:\setminus\{t\}, :})$
            \EndFor\label{line:pullKVEnd}
            \For{$k \gets 1, \dots, N - 1$}\Comment{\emph{Push O}}\label{line:pushOCommInterBegin}
                \State $t' \gets (t-k)\% N$
                \State $\textsc{ScatterPush}(\{O_{t', :}\}, \{O^\text{buf}_{t, u}\}, (t', :, r))$\label{line:commOuter4}
            \EndFor\label{line:pushOCommInterEnd}
            \State $\textsc{RingAttn}(Q_{t, :}, K_{:\setminus\{t\}, :}, V_{:\setminus\{t\}, :}, O_{t, :}, l_{t, :}, m_{t, :})$\label{line:pushOCompute}
            \State $\textsc{ScatterPush}(\{O_{t, :}\}, \{O^\text{buf}_{t, u}\}, (t, :, r))$\Comment{For Ulysses Attention}\label{line:pushOCommIntra}\label{line:commOuter5}
            \State {$\textsc{BarrierAll}()$}\label{line:barrierEnd}
            \State \Return $O$
        \EndFunction{}
    \end{algorithmic}
\end{algorithm}
\system first applies a \textsc{ScatterPush} operation to perform the all-to-all operators on the local tensors $Q_t$, $K_t$, and $V_t$ for Ulysses Attention (Line \ref{line:ulyssesBegin}).
Then, \system synchronizes all GPUs with the ``\emph{BarrierAll}'' primitive to ensure previous \textsc{ScatterPush} operations are completed (Line \ref{line:barrierBegin}).
Next, \system issues all \textsc{GatherPull} operations needed for both the \emph{Pull Q} and \emph{Pull KV} stages (Line \ref{line:issueCommBegin}-\ref{line:issueCommEnd}).
By issuing all pulls without waiting for each to complete, we maximize overlap between communication and computation.

Next, \system starts the attention computation.
\system first performs the Ring Attention for the first \emph{Pull Q} stage using the local $Q_{t, t}$, $K_{t, t}$, and $V_{t, t}$ tensors (Line \ref{line:firstPullQ}).
Then, in both \emph{Pull Q} and \emph{Pull KV} stages, \system waits for the communication of each chunk of the $Q$ and the $KV$ tensors to finish before performing the attention computation using the newly received chunk (Line \ref{line:pullQBegin}-\ref{line:pullQEnd} and Line \ref{line:pullKVBegin}-\ref{line:pullKVEnd}).
For the attention computation in both stages, \system invokes the \textsc{RingAttn} function (Line \ref{line:ringAttnBegin}-\ref{line:ringAttnEnd}) that implements Ring Attention using one-sided communication.
Instead of using a ring-like communication, we directly pull the $KV$ tensors from the other GPUs given the rank with the $\textsc{Pull}$ primitive, thus removing the unnecessary synchronization needed in the original Ring Attention.
Note that in each Pull KV stage, after waiting for the communication of the $KV$ tensors to finish, we need to perform a ``\emph{barrier}'' operation to synchronize all GPUs that are in the Ring Attention process group (Line \ref{line:pullKVBarrier}), since Ring Attention requires the $KV$ tensors on all participating GPUs to be ready before starting the attention computation.

Finally, the output tensor $O_{:\setminus\{t\}, :}$ is ready to be sent from GPU $t$ to the other GPUs.
Since we need to account for the Ulysses Attention's all-to-all operator on the output tensor $O$, \system issues a \textsc{ScatterPush} operator to send $O_{:\setminus\{t\}, :}$ to the other GPUs in the \emph{Push O} stage (Line \ref{line:pushOCommInterBegin}-\ref{line:pushOCommInterEnd}).
At the same time, \system computes the remaining output tensor $O_{t, :}$ locally on GPU $t$ (Line \ref{line:pushOCompute}).
After finishing the computation of $O_{t, :}$, \system performs another \textsc{ScatterPush} operator to send $O_{t, :}$ to the other GPUs for Ulysses Attention (Line \ref{line:pushOCommIntra}).
Finally, \system performs another ``\emph{BarrierAll}'' operation to ensure all communication operators are finished before returning the output tensor (Line \ref{line:barrierEnd}).

In this way, \system leverages one-sided communication libraries to implement \torusShort, Ulysses, and Ring Attention while still overlapping the communication with computation and reducing the synchronization overheads.

\section{Evaluation}
\label{sec:evaluation}

\begin{figure*}[t]
    \centering
    \includegraphics[width=\linewidth]{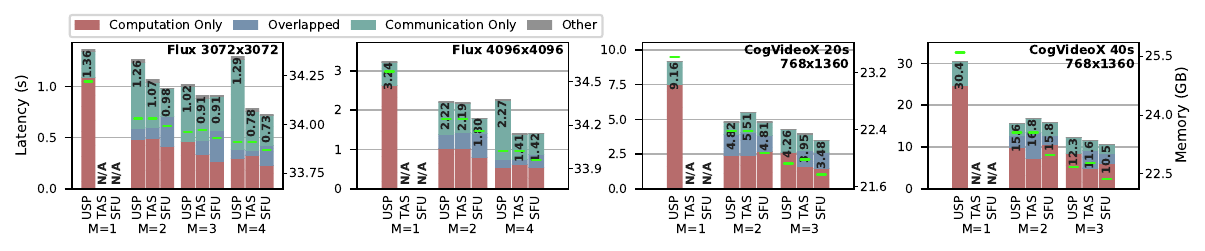}
    \caption{Overall end-to-end performance of \system compared to baselines with their own optimal distributed configurations, where $M$ stands for the number of GPU machines.}
    \label{fig:overall-performance}
\end{figure*}

\subsection{Experiment Setups}

\textbf{Baselines}~ We compare \system, denoted by \systemShort, with the state-of-the-art work USP \cite{usp}.
We also evaluate a variant of \system that only applies the topology-aware sequence parallelism, denoted by \tasShort, to demonstrate the benefits of \torus and one-sided communication implementation.
Note that when using a single GPU machine, all methods degrade to Ulysses Attention since there is no inter-machine communication.

\noindent\textbf{Hardware \& Software}~ We conduct all experiments on four AWS \texttt{p4de.24xlarge} instances, each equipped with 8 NVIDIA A100 GPUs (40 GiB per GPU) interconnected via NVSwitch. We also enable the AWS Elastic Fabric Adapter (EFA), offering up to 400 Gbps bandwidth for efficient inter-machine communication.
\system and all baselines are executed with NVIDIA driver version 570.172.08, CUDA 12.8.0, PyTorch 2.8.0, NCCL 2.27.3, and NVSHMEM 3.4.5 installed on each machine.

\noindent\textbf{Models \& Workloads}~
To prove the generalizability of \system, we use four commonly-used workloads for evaluation.
For image generation, we use Flux \cite{flux-image-gen} with 12B parameters for generating 3072$\times$3072 and 4096$\times$4096 images. 
For video generation, we use CogVideoX \cite{cogvideox} with 5B parameters for producing 20 or 40 seconds 768$\times$1360 videos.
Both models use 24 attention heads, with head dimensions of 128 and 64 for Flux and CogVideoX, respectively.

\subsection{Overall Performance}
\label{sec:overall-performance-exp}

\begin{figure*}[t]
    \centering
    \includegraphics[width=\linewidth]{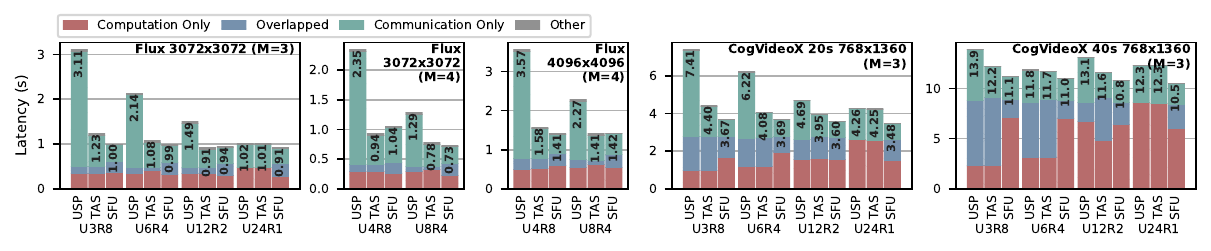}
    \caption{Overall end-to-end performance on various distributed configurations with 4 or 3 GPU machines.}
    \label{fig:overall-performance-dist}
\end{figure*}

\textbf{Optimal Distributed Configuration}~ Figure \ref{fig:overall-performance} shows the overall inference latency and memory consumption of one sampling step between \system and prior state-of-the-art works, with their own optimal distributed configurations.
Note that for different workloads, we benchmark different sets of the GPU machine numbers.
This is because SP techniques require the sequence length and number of heads to be divisible by the number of GPU machines.
Moreover, we only show USP performance numbers for $M = 1$ cases since all methods will degrade to Ulysses Attention when using a single GPU machine and \system's optimizations are tailored to inter-machine communication.
We draw four conclusions from the figure. 
First, we demonstrate that \tasShort performs worse than USP when using 2 GPU machines.
This is because when using 2 GPU machines, \tasShort introduces the same inter-machine communication volume as USP but its communication latency is not overlapped with computations.
Second, we find that \tasShort achieves a speedup by \Eval{avg:method=tas} on average (up to \Eval{max:method=tas}) when comparing to USP when using more than 2 GPU machines for distributed DiT inference.
This is because \tasShort introduces a much lower inter-machine communication volume than USP with more than 2 GPU machines.
Third, we show that by overlapping all-to-all operators with computation and leveraging one-sided communication libraries, \system achieves a better speedup than \tasShort by \Eval{avg:method=sfu} on average (up to \Eval{max:method=sfu}).
This is because \system is able to overlap inter-machine communication with computation.
Fourth, we find that \system does not introduce higher memory consumption compared to prior state-of-the-art work USP \cite{usp}, thanks to its careful design and implementation with one-sided communication library in Algorithm \ref{alg:1-side-comm}, where we use at most one copy buffer of $Q$, $K$, $V$, and $O$ tensors.
Overall, we conclude that \system has lower inter-machine communication overheads compared to prior state-of-the-art frameworks while introducing negligible memory overheads, and that all proposed key ideas are required to achieve the optimal system performance.

\noindent\textbf{Other Distributed Configurations}~ Figure \ref{fig:overall-performance-dist} compares the system performance of \system with baselines at different distributed configurations, where U$x$R$y$ stands for using Ulysses Attention with parallelization degree $P_u = x$ and Ring Attention with parallelization degree $P_r = y$.
We make three observations from the figure.
First, we found that both our work \tasShort and \systemShort can consistently outperform USP on all setups by \Eval{avg:diffdist_method=tas} speedup on average (up to \Eval{max:diffdist_method=tas}) and \Eval{avg:diffdist_method=sfu} speedup on average (up to \Eval{max:diffdist_method=sfu}), respectively.
This is because \tasShort and \systemShort conduct Ring Attention within a single GPU machine and hence significantly reduce inter-machine communication volume.
Second, we found that for all methods, increasing the Ulysses Attention degree $U$ can always achieve better speedups, except for \tasShort where U12R2 is better than U24R1.
This is because increasing Ulysses Attention degree $U$ reduces the Ring Attention degree $R$.
Since Ring Attention divides one attention operator along the sequence length into multiple ones, choosing a small $R$ prevents the fragmentation of attention operations on each GPU, which reduces kernel launching overheads, avoids bubbles in GPU warp schedulers, and thus better utilizes the GPU computational resources.
However, for \tasShort since the all-to-all operators are not overlapped with communication, excessively using Ulysses Attention could result in larger non-overlapping communication overheads.

\subsection{Layerwise Performance}

\begin{figure}[t]
    \centering
    \includegraphics[width=.887\linewidth]{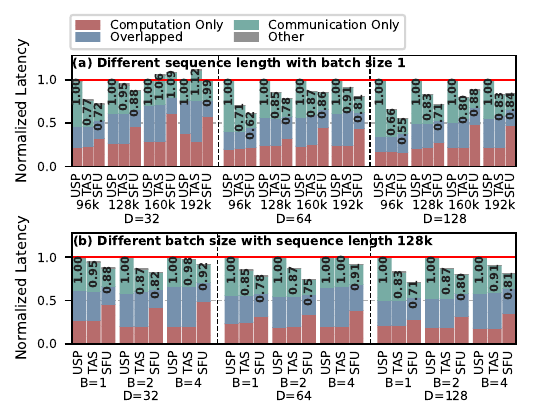}
    \vspace{-10pt}
    \caption{Normalized latency of \system at different head dimensions, when varying sequence length and batch sizes.}
    \label{fig:bs_ql}
\end{figure}

We conduct micro-benchmarks by varying the sequence length, head dimension, and batch size, and measuring the latency of a single attention layer to understand how \system's performance generalizes to different workloads.

\noindent\textbf{Different Head Dimensions}~
Figure \ref{fig:bs_ql}a and Figure \ref{fig:bs_ql}b demonstrate the performance of \system compared against USP with different head dimensions ($D$).
We observe that for all head dimensions, \system always achieves a better performance compared to USP, by \Eval{avg:D=32_method=sfu}, \Eval{avg:D=64_method=sfu}, and \Eval{avg:D=128_method=sfu} speedup on average for $D = 32, 64, 128$ respectively.
Besides, we notice that as the head dimension increases, \system provides better speedups.
This is because sharding attention computation makes the kernel launching grids fairly small, which underutilizes GPU computation resources.
By increasing the head dimension, we make the attention computation easier to saturate the GPU computational resources.

\noindent\textbf{Different Sequence Length}~
Figure \ref{fig:bs_ql}a demonstrates the normalized latency of \system compared to USP with different sequence lengths. 
We draw two observations from the figure.
First, we find that \system almost consistently provides speedup compared to USP, demonstrating \Eval{avg:L=96k_method=sfu}, \Eval{avg:L=128k_method=sfu}, \Eval{avg:L=160k_method=sfu}, and \Eval{avg:L=192k_method=sfu} speedup on average with sequence length 96k, 128k, 160k, and 192k, respectively.
The only exception case is when we have more than 160k tokens when the number of head dimensions $D = 32$.
Second, we notice that as the sequence length increases, \system achieves less speedup compared to USP.
This is because as the sequence length increases, the computation on each GPU increases quadratically while the communication on each GPU only increases linearly, which thus makes the attention layer execution compute-bounded and offsets the benefits of communication overhead reduction techniques introduced by \system.

\noindent\textbf{Different Batch Size}~
Figure \ref{fig:bs_ql}b shows the normalized latency of \system compared to USP with different batch sizes.
We observe that \system consistently outperforms USP on all batch sizes, demonstrating \Eval{avg:B=1_method=sfu}, \Eval{avg:B=2_method=sfu}, \Eval{avg:B=4_method=sfu} speedup on average.
However, we notice that as the batch size increases, there is no clear trend on the speedup of \system compared to USP across various head dimensions.

Overall, we conclude that \system can provide consistent speedups over prior state-of-the-art work across different layerwise configurations, which demonstrates the generalizability of \system to various workloads.

\section{Related Works}
\label{sec:related-works}

\textbf{Distributed Inference of DiTs}~
DiT inference is both compute- and memory-intensive, and distributed execution is often needed to avoid out-of-memory errors, reduce latency, and improve hardware utilization.
Ulysses Attention \cite{ulysses}, Ring Attention \cite{ring-attn}, and their combination USP \cite{usp} are the main SP techniques used today.
DistriFusion \cite{distrifusion} and PipeFusion \cite{pipefusion} explore lossy communication-hiding techniques that leverage temporal redundancy across sampling steps to hide communication latency with little accuracy loss.
However, these approaches are limited by their focus on specific network topologies or their lack of addressing network constraints, resulting in significant communication overhead..

\noindent\textbf{Communication-Computation Overlapping}~
We classify prior works for overlapping communication and computation into manual and compiler-based approaches.
For manual approaches, Flux \cite{flux-comm-overlap} hides the latency of \texttt{AllGather} and \texttt{ReduceScatter} operators with computation in tensor parallelism for LLMs.
DeepEP \cite{deepseek-v3,deepep} uses hand-optimized kernels that start communication from inside CUDA kernels for fine-grained overlap. 
Comet \cite{comet} hides all-to-all latency but focuses on GEMM operators and does not directly support attention. 
ScaleFusion \cite{scalefusion} breaks all-to-all operators to overlap with computation but targets diffusion transformers with spatial-temporal attention architecture.
However, these works cannot be directly applied to DiT inference with general attention architecture, which has its unique challenges as discussed in Section \ref{sec:motivation}.
For compiler-based approaches, frameworks like TileLink \cite{tilelink}, Triton-Distributed \cite{triton-distributed}, Mercury \cite{mercury} integrate communication primitives into compilers (e.g., Triton) to automatically overlap communication and computation.
While these works provide general abstractions for communication-computation overlapping, directly applying them to discover the specialized overlapping techniques such as \torus is non-trivial since it requires careful scheduling of communication and computation operators to ensure fine-grained overlapping.
We believe \system can be integrated into these compiler-based works to further enhance the performance of DiT inference.

\section{Conclusions}
\label{sec:conclusion}
In this paper, we present \system, a novel distributed DiT inference engine.
To address the key challenges in existing SP techniques, \system introduces three key techniques: (1) topology-aware sequence parallelism that combines Ulysses Attention and Ring Attention in a topology-aware manner; (2) \torus that pipelines inter-machine communication with computation to hide communication overheads; and (3) one-sided communication implementation to reduce synchronization overheads.
Our extensive experiments demonstrate that \system outperforms existing SP techniques with \Eval{avg:method=sfu} speedup on average (up to \Eval{max:method=sfu}).
We hope \system provides new insights for efficient and scalable distributed DiT inference on modern GPU machines.

\begin{acks}
This paper is supported by Vector Institute Research grants, the Canada Foundation for Innovation JELF grant, NSERC Discovery grant, AWS Machine Learning Research Award, Facebook Faculty Research Award, Google Scholar Research Award, and VMware Early Career Faculty Grant.
\end{acks}

\bibliography{main}
\bibliographystyle{plain}
\balance

\clearpage
\newpage

\appendix
\section{One-Sided Communication Implementation Details}

In \system, we use four communication functions in our implementation (Algorithm \ref{alg:1-side-comm}).
\textsc{ScatterPush} and \textsc{GatherPull} are used to implement the all-to-all operators in Ulysses Attention, which rely on \texttt{nvshmemx\_putmem\_on\_stream} and \texttt{\seqsplit{nvshmemx\_getmem\_on\_stream}} APIs in NVSHMEM, respectively.
\textsc{ScatterPush} pushes data from multiple source buffers to the destination buffer on different remote GPUs specified by the $ranks$ list, while \textsc{GatherPull} pulls data from multiple destination buffers multiple destination buffers.
\textsc{Barrier} and \textsc{BarrierAll} are used to synchronize data consistency across multiple GPUs, which relies on \texttt{\seqsplit{nvshmemx\_barrier\_on\_stream}} and \texttt{nvshmem\_barrier\_all\_on\_stream} APIs respectively.
\textsc{Barrier} synchronizes GPUs in a specified process group, while \textsc{BarrierAll} synchronizes all GPUs.

Note that for all communication functions, we pass in a CUDA stream parameter to ensure the communication operations are executed in the correct order with computation operations.
In \system's implementation, we create two CUDA streams aside from the default stream (used for computation), one for the Ring Attention communication, denoted by $stream_\text{ring}$ and one for the other communications, denoted by $stream_\text{other}$.
Specifically, we execute all communication operations in the $\textsc{RingAttn}$ (Line \ref{line:ringAttnBegin}-\ref{line:ringAttnEnd}) function on the $stream_\text{ring}$ and the other communication operations in $stream_\text{other}$ (Line \ref{line:commOuter1}, \ref{line:commOuter2}, \ref{line:commOuter3}, \ref{line:commOuter4}, and \ref{line:commOuter5}).
In this way, both intra- and inter-machine communication latencies are hidden by computation operators with the one-sided communication library in \system.

\section{Ablation Studies}
\begin{figure*}[t]
    \centering
    \includegraphics[width=\linewidth]{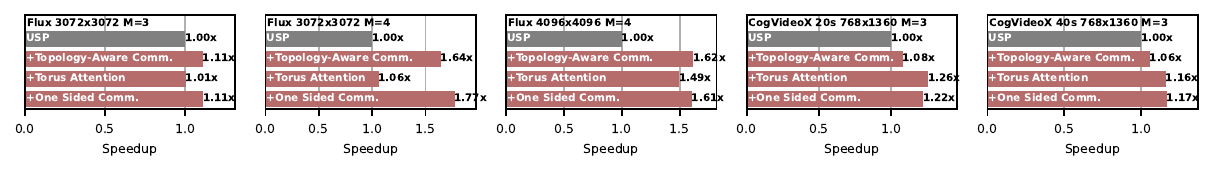}
    \caption{Ablation studies in \system.}
    \label{fig:ablation}
\end{figure*}

Figure \ref{fig:ablation} shows the end-to-end normalized latency of one sampling step of all workloads used in Section \ref{sec:evaluation} when gradually adding each proposed method.
We observe that the topology-aware communication scheduling (\tasShort) consistently outperforms USP \cite{usp} by \Eval{avg:method=tas} speedup on average (up to \Eval{max:method=tas}).
However, for the \torus and one-sided communication implementation, their benefits are workload-dependent.
For image generation workloads with Flux models, \torus implemented with NCCL does not provide speedups compared to \tasShort.
This is because image generation workloads have relatively short sequence lengths and thus the communication volume of Ulysses Attention is not large enough to become the performance bottleneck.
However, when implemented with one-sided communication libraries, \torus still provides significant speedups compared to \tasShort, since the one-sided communication implementation can better overlap communication with computation.
For video generation workloads, \torus itself can already provide significant speedups compared to \tasShort, since video generation workloads have long sequence lengths and thus the communication volume of Ulysses Attention becomes the performance bottleneck.
However, the one-sided communication implementation provides marginal speedups for video generation workloads, since the communication overheads are already largely hidden by \torus.

Overall, we conclude that all three proposed methods in \system are required to achieve the optimal system performance.

\section{Merging Attention Outputs with Different $KV$ Tensors}

In both Ring Attention and \torus, the $KV$ tensors are distributed among multiple GPUs. 
Thus, during attention computation, each GPU will compute the attention outputs using different $KV$ tensors.
To correctly compute the final attention output, we need to merge the attention outputs computed with different $KV$ tensors, similar to FlashAttention \cite{flash-attn,flash-attn-2,flash-attn-3}.
In this section, we describe how to merge the attention outputs computed with different $KV$ tensors.

\textbf{Algorithm for Merging Attention Outputs}~
Formally, suppose that we have split the $KV$ tensors along the sequence length into $N$ partitions, and denote the $i$-th partition by $K_i, V_i$, the final attention output of a query tensor $Q$ on each $KV$ partition can be computed as:
\begin{equation}\label{eq:attn}
    \begin{aligned}[c]
        &m_i = \textrm{rowmax}(QK_i^\intercal)\\
        &l_i = \textrm{rowsum}(e^{QK_i^\intercal - m_i})\\
        &O_i = (e^{QK_i^\intercal - m_i}V_i)/l_i
    \end{aligned}
\end{equation}
Thus, we can represent the computation on each GPU $i$ as a triplet of $(O_i, l_i, m_i)$, denoted by an intermediate result matrix $A_i = \begin{bsmallmatrix}
        O_i \\
        l_i \\
        m_i
    \end{bsmallmatrix}$.
To combine the attention outputs computed on two different $KV$ tensors, i.e., $A_i$ and $A_j$, we can merge them as $
    A_i \oplus A_j =
    \begin{bsmallmatrix}
        O \\
        l \\
        m
    \end{bsmallmatrix}
$, where,
\begin{equation}
    \begin{aligned}[c]
        &m = \max(m_i, m_j) \\
        &l = l_ie^{m_i - m} + l_je^{m_j - m} \\
        &O = (O_i l_i e^{m_i - m} + O_j l_j e^{m_j - m}) / l
    \end{aligned}
\end{equation}
By applying the above merging operation on all $N$ partitions, we can obtain the sum of all $A_i$ matrices, $A_1 \oplus \cdots \oplus A_N$, and then extract the first tensor $O$ from the resulting matrix.

\textbf{Optimizing Floating-Point Operations}~
To avoid expensive floating-point divisions when merging the attention outputs, we can rewrite each intermediate result matrix $A_i$ as $A_i'$, similar to FlashAttention-2 \cite{flash-attn-2}, we have,
$
    A_i' = \begin{bsmallmatrix}
        O'_i \\
        l_i \\
        m_i
    \end{bsmallmatrix}
$,
where $O'_i = O_il_i$.
By doing so, we can rewrite the merging operation as,
$
    A_i' \oplus A_j' =
    \begin{bsmallmatrix}
        O' \\
        l \\
        m
    \end{bsmallmatrix}
$,
where,
\begin{equation}
    O' = O'_i e^{m_i - m} + O'_j e^{m_j - m}
\end{equation}
After merging all intermediate result matrices, we can finally compute the final output tensor $O$ as $O = O'/l$.
In this way, we only need to perform one floating-point division at the end of the merging process, thus avoiding unnecessary floating-point operations.
\begin{algorithm}[t]\scriptsize
    \caption{\scriptsize Attention kernel with multiple $Q$ and $KV$ tensors on Ampere GPUs}
    \label{alg:merge-kv}
    \begin{algorithmic}[1]
        \Require 
            \item[] Batch size $B$, Number of heads $H$, Head dimension $D$, Number of $Q$ and $KV$ tensors $nQO$, $nKV$
            \item[] $nQO$ tensors $\{Q_i\}_{i=1}^{nQO}$ of shape $[B, lQO_i, H, D]$
            \item[] $nKV$ tensors $\{K_i\}_{i=1}^{nKV}$ of shape $[B, lKV_i, H, D]$
            \item[] $nKV$ tensors $\{V_i\}_{i=1}^{nKV}$ of shape $[B, lKV_i, H, D]$
            \item[] $nQO$ tensors $\{O_i\}_{i=1}^{nQO}$ of shape $[B, lQO_i, H, D]$
            \item[] $nQO$ tensors $\{l_i\}_{i=1}^{nQO}$ of shape $[B, H, lQO_i]$
            \item[] $nQO$ tensors $\{m_i\}_{i=1}^{nQO}$ of shape $[B, H, lQO_i]$
            \item[] Tile size $tQO$ along the sequence length dim. of $Q$ tensors
            \item[] Tile size $tKV$ along the sequence length dim. of $KV$ tensors      
            \item[] Whether to finalize the output tensors \texttt{finalize}
        \State $cQO_i \gets \sum_{j = 0}^{i-1} \lceil \frac{lQO_j}{tQO} \rceil$ (On CPU)\label{line:computeCQO}
        \State Launch CUDA kernel with grid size $[\sum_i \lceil \frac{lQO_i}{tQO} \rceil, B, H]$\label{line:launch}
        \For{each CUDA thread block}
            \State Declare shared memory (SMem) $sQ$, $sK$, $sV$
            \State Declare register memory (RMem) $rO$, $rM$, $rL$
            \State $b, h \gets \texttt{blockIdx.y}, \texttt{blockIdx.z}$
            \State Find $i$, s.t. $cQO_{i} \leq \texttt{blockIdx.x} < cQO_{i + 1}$\label{line:findq}
            \State $q \gets tQO \cdot (\texttt{blockIdx.x} - cQO_{i})$\label{line:computeTileq}
            \State $\textsc{GMemToSMem}(Q_i[b, q:q+tQO, h, :] \rightarrow sQ)$\label{line:loadQ}
            \State $\textsc{GMemToRMem}(O_i[b, q:q+tQO, h, :] \rightarrow rO)$
            \State $\textsc{GMemToRMem}(m_i[b, h, q:q+tQO] \rightarrow rM)$\label{line:loadMLBegin}
            \If {{\color{blue}$\texttt{threadIdx.x} \% 4 = 0$}}\label{line:conditionLLoad}
                \State $\textsc{GMemToRMem}(l_i[b, h, q:q+tQO] \rightarrow rL)$
            \Else
                \State $rL \gets 0$
            \EndIf\label{line:loadMLEnd}
            \For{$j \gets 1, 2, \dots, nKV$}
                \For {each tile $k$ on sequence length dim. of $K_j$ and $V_j$}
                    \State $\textsc{GMemToSMem}({K_j}[b,k:k+tKV,h,:] \rightarrow sK)$\label{line:loadK}
                    \State $\textsc{GMemToSMem}({V_j}[b,k:k+tKV,h,:] \rightarrow sV)$\label{line:loadV}
                    \State $rS \gets \textsc{MatMul}_\text{m16n8k16}(sQ, sK)$\label{line:computeBegin}\label{line:qxkt}
                    \State $rM'\gets\textsc{WarpReduceMax}\texttt{<4>}(\max(rM, \mathrm{rowmax}(rS)))$\label{line:warpreducemax}
                    \State $rP \gets e^{\frac{rS - rM'}{\sqrt{D}}}$
                    \State $rL \gets rL \cdot e^{\frac{rM - rM'}{\sqrt{D}}} + \mathrm{rowsum}(rP)$
                    \State $rO \gets rO \cdot e^{\frac{rM - rM'}{\sqrt{D}}}+ \textsc{MatMul}_\text{m16n8k16}(rP, sV)$\label{line:computeEnd}\label{line:pxvt}
                    \State $rM \gets rM'$
                \EndFor
            \EndFor
            \State $rL \gets \textsc{WarpReduceSum}\texttt{<4>}(rL)$\label{line:warpreducesum}
            \If {finalize}
                \State $rO \gets rO / rL$\label{line:writebackO}
            \Else
                \State $\textsc{RMemToGMem}(rL \rightarrow l_i[b, h, q:q+tQO])$
                \State $\textsc{RMemToGMem}(rM \rightarrow m_i[b, h, q:q+tQO])$
            \EndIf
            \State $\textsc{RMemToGMem}(rO \rightarrow O_i[b, q:q+tQO, h, :])$
        \EndFor
    \end{algorithmic}
\end{algorithm}

\textbf{Customized CUDA Kernel Implementation}~
\torus requires to compute attention of multiple $Q$ and multiple $KV$ tensors, since each received $QKV$ tensors can be discontinuous in memory.
However, directly launching separate CUDA kernels for each attention computation and the merging operation causes high kernel launching overheads and expensive global memory accesses.
To avoid this, we fuse the attention computation for multiple $QKV$ tensors and the merging operation in a single kernel.

Algorithm \ref{alg:merge-kv} shows the pseudocode of our custom CUDA kernel implementation on NVIDIA Ampere architecture, which we implement with the CUTLASS CuTe library \cite{cutlass} similar to FlashAttention-2 \cite{flash-attn-2}, and our design can be easily extended to other GPU architectures.
Specifically, suppose we have $nQO$ $Q$ tensors and $nKV$ $KV$ tensors.
The shape of $i$-th $Q$ tensor and the $j$-th $KV$ tensors are $[B, lQO_i, H, D]$ and $[B, lKV_j, H, D]$ respectively.
We can launch a CUDA kernel with a 3-D grid whose shape is $[\sum_i \lceil \frac{lQO_i}{tQO} \rceil, B, H]$ (Line \ref{line:launch}), where $tQO$ is the sequence length tile size of the $Q$ tensors, i.e., the partition size of the sequence dimension of the $Q$ tensors that is handled in each CUDA thread block.
Similar to FlashAttention, each CUDA thread block will first load a tile of $Q$ tensor and a tile of $KV$ tensor into the shared memory (Line \ref{line:loadQ}, \ref{line:loadK}, and \ref{line:loadV}), and then compute the attention outputs on the loaded tiles (Line \ref{line:computeBegin}-\ref{line:computeEnd}).
However, since we have multiple $Q$ tensors, we need to first locate which $Q$ tensor the current thread block is handling based on its block index in the first dimension.
To achieve this, we first compute a auxiliary array $\{cQO_i\}_{i=0}^{nQO + 1}$ that is the exclusive prefix sum of the number of tiles for each $Q$ tensor, i.e., $cQO_i = \sum_{j = 0}^{i - 1} \lceil \frac{lQO_j}{tQO} \rceil$ (Line \ref{line:computeCQO}).
Then inside each CUDA thread block we binary search the x-axis of the current block index, i.e., \texttt{blockIdx.x}, to find the corresponding $Q$ tensor index $i$ such that $cQO_{i} \leq \texttt{blockIdx.x} < cQO_{i + 1}$ (Line \ref{line:findq}).
We can then find the tile index of the $i$-th $Q$ tensor as $\texttt{blockIdx.x} - cQO_i$ and compute attention outputs similarly as FlashAttention (Line \ref{line:computeTileq}).

For attention computation, we use the PTX instruction \texttt{\seqsplit{mma.sync.aligned.m16n8k16}} to perform the matrix multiplications $Q \times K^\intercal$ (Line \ref{line:qxkt}) and $P \times V^\intercal$ (Line \ref{line:pxvt}).
Figure \ref{fig:mma} (left) shows how the thread block is composed, where we use 8 warps in each thread block.
Each warp is responsible for computing one row of the result matrix by iterating through all columns of the result matrix.
For each iteration, the layout of which is shown in Figure \ref{fig:mma} (right), where we use two of the same PTX instructions. 
This is because we use the PTX instruction for vectorized memory load, \texttt{ldmatrix.sync.aligned.m8n8.x4}, for loading from shared memory, which maximizes the memory bandwidth utilization on NVIDIA Ampere GPUs, i.e., each thread reading a 128-bit data.

To handle multiple $KV$ tensors and fuse the merging operation, instead of initializing the auxiliary output tensor $O'$, the running maximum tensor $m$, and the running sum tensor $l$, to their default values (i.e., zeros, negative infinity, and zeros, respectively) at the beginning of the CUDA kernel, we load the $m$ and $l$ from global memory that are computed from the previous $KV$ tensors (Line \ref{line:loadMLBegin}-\ref{line:loadMLEnd}).
Since we need to compute the rowsum and rowmax as indicated by Equation \ref{eq:attn} of the $Q\times K^\intercal$ matrix, according to the result matrix layout mentioned above, we need to first conduct local rowsum and rowmax and then perform warp-level reductions across every 4 threads using warp-level shuffle primitives (Line \ref{line:warpreducemax} and \ref{line:warpreducesum}).
However, this could lead to incorrect $rL$ values if each thread loads the same $rL$ values from the global memory, since if we do not have any available $KV$ tensors to process, the same $rL$ values loaded from the global memory will still be summed up across every 4 threads at the end of the kernel and its values will be 4 times the original values.
To avoid this, we only have one out of every four threads to load the $rL$ values from the global memory, and let other threads initialize their local $rL$ values to zeros (Line \ref{line:conditionLLoad}) and our implementation simply chooses the thread with $\texttt{threadIdx.x} \% 4 = 0$.
Note that loading the running maximum tensor into $rM$ does not have this issue since the maximum operation is idempotent.
After processing the last $KV$ tensors (e.g., the last step in Ring Attention), we can then finalize the output tensor (i.e., when \texttt{finalize} is true) by dividing $O'$ by the $l$ tensor (Line \ref{line:writebackO}).

\begin{figure}[t]
    \centering
    \includegraphics[width=\linewidth]{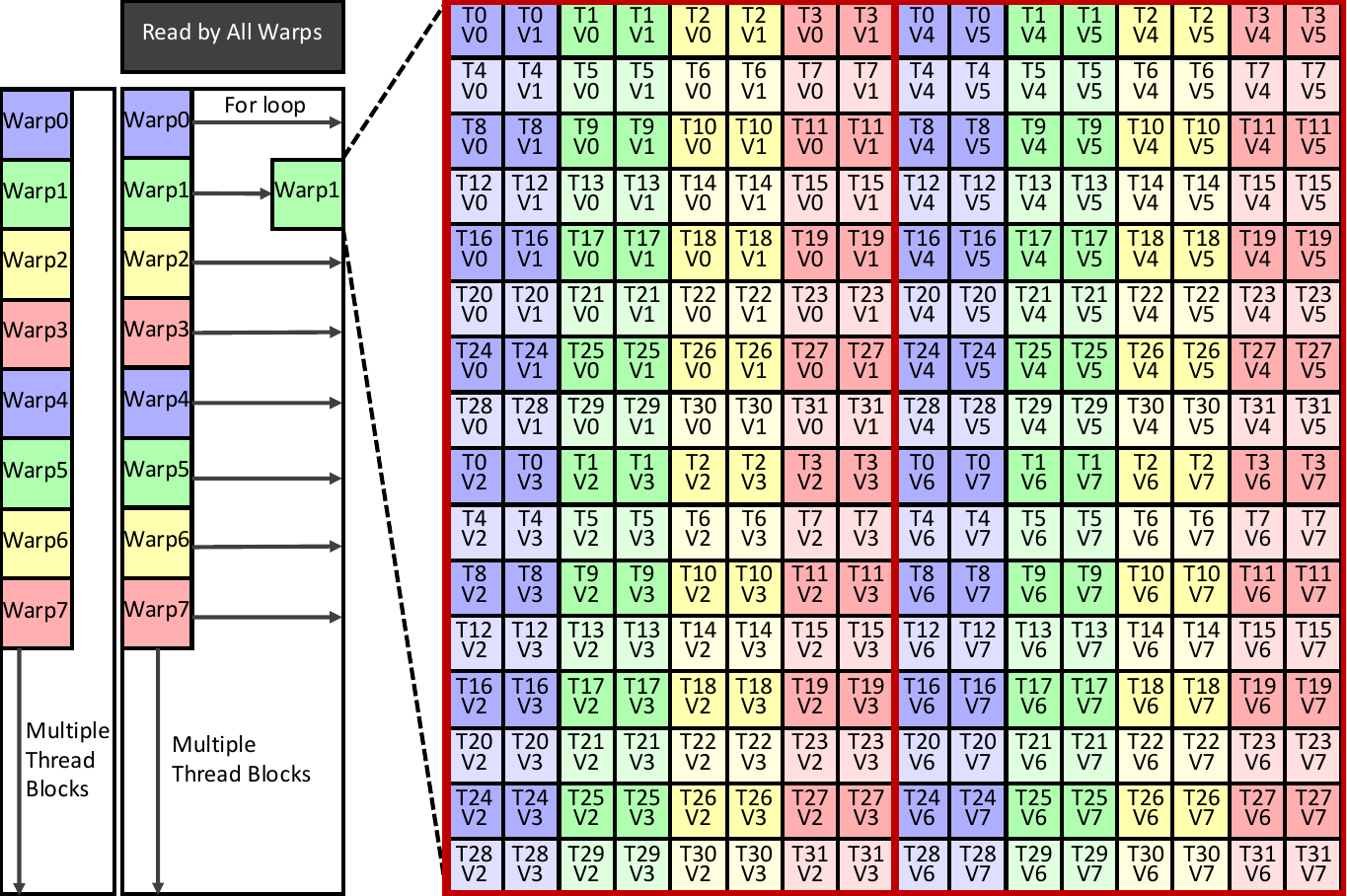}
    \caption{The block-level and thread-level layout of the matrix multiplication used for $Q \times K^\intercal$ and $P \times V^\intercal$ with the PTX instruction \texttt{mma.sync.aligned.m16n8k16} (shown in red box).}
    \label{fig:mma}
\end{figure}

\begin{figure}[t]
    \centering
    \includegraphics[width=\linewidth]{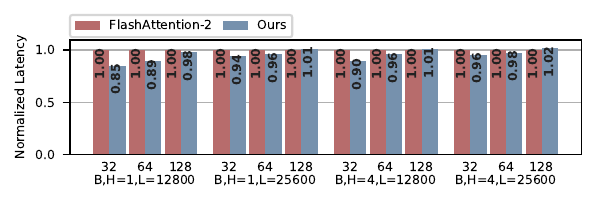}
    \caption{Normalized latency of our customized CUDA kernel and FlashAttention-2 kernels \cite{flash-attn-2} on A100 GPUs.}
    \label{fig:kernel}
\end{figure}

Figure \ref{fig:kernel} shows the performance comparison of our customized CUDA kernel against FlashAttention-2 \cite{flash-attn-2} with a single $QKV$ tensor, where we launch the customized CUDA kernel with 256 threads for each CUDA thread block and use the same tile size as FlashAttention-2.
We conclude that our customized CUDA kernel introduces negligible performance overheads compared to FlashAttention-2 even with the ability to process multiple $Q$ and $KV$ tensors and the merging operation.
\section{Inter-Machine Communication Volume Analysis}\label{app:complexity-analysis}

Suppose we have $N$ GPU machines, each with $M$ GPUs, and the parallelization degree of Ulysses and Ring Attention are $P_u$ and $P_r$, respectively.
We denote the batch size as $B$, the sequence length across all GPUs as $L$, the number of heads as $H$, and the head dimension as $D$.
Thus, each \emph{GPU machine} has a sequence length of $L / N$.

\noindent\textbf{USP.}~When $P_r \geq N$, USP performs inter-machine communication using Ring Attention.
Thus, the inter-machine communication volume (i.e., the number of elements) per GPU is:
\begin{equation}
    \text{V}_{\text{USP}} = 2 \cdot (N - 1) \cdot \frac{BLHD}{N}
\end{equation}
When $P_r \leq N$, USP performs inter-machine communication in both Ring and Ulysses Attention, and the Ulysses Attention uses $N / P_r$ as its parallelization degree for inter-machine communication.
Thus, in this case, USP incurs an inter-machine communication volume of:
\begin{equation}
    \begin{split}
    \text{V}_{\text{USP}} &= \left(2 \cdot(P_r - 1) \cdot \frac{N}{P_r}+ 4 \cdot \frac{N / P_r - 1}{N / P_r} \right) \frac{BLHD}{N}\\
    &= \left(2N + 4 - \left(2\cdot \frac{N}{P_r} + 4 \cdot \frac{1}{N / P_r}\right)\right)\frac{BLHD}{N}\\
    &\leq \left(2N + 4 - 6\right)\frac{BLHD}{N}\quad\text{(by setting $P_r = N$)}\\
    &= \left(2N - 2\right)\frac{BLHD}{N}
    \end{split}
\end{equation}

\noindent\textbf{\system.}~We can analyze the inter-machine communication volume of \system in a similar way.
When $P_u \geq N$, \system performs inter-machine communication using Ulysses Attention, and the inter-machine communication volume per GPU is:
\begin{equation}
    \text{V}_\text{\systemShort} = 4 \cdot \frac{N - 1}{N} \cdot \frac{BLHD}{N}
\end{equation}
When $P_u \leq N$, Ring Attention in \system also causes inter-machine communication, and its parallelization degree is $N / P_u$.
Thus, in this case \system incurs an inter-machine communication volume of:
\begin{equation}
    \begin{split}
    \text{V}_\text{\systemShort} &= \left(2 \cdot {(N / P_u - 1)} + 4 \cdot \frac{P_u - 1}{P_u} \cdot \frac{N}{P_u} \right) \frac{BLHD}{N}\\
    &= \left(\left(6 - 4\cdot\frac{1}{P_u}\right)\frac{N}{P_u} - 2\right)\frac{BLHD}{N}\\
    &\geq 4\cdot \frac{N - 1}{N}\cdot \frac{BLHD}{N}\quad\text{(by setting $P_u = N$)}\\
    \end{split}
\end{equation}

Thus, to show that \system can always introduce less inter-machine communication than USP, we only need to show that this is the case when $P_u \leq N$ and $P_r \leq N$.
In this case, we have $P_r = \frac{NM}{P_u} \leq N$, which implies that $P_u \geq M$.
The following lemma proves that \system always incurs less inter-machine communication volume than USP.

\begin{lemma}
Let $\text{V}_\text{diff} = \frac{\text{V}_\text{USP} - \text{V}_\text{\systemShort}}{BLHD/N} = \frac{4N}{P_u^2} - \frac{4M+ 6N}{P_u} - \frac{2P_u}{M} + 2N + 6$,
then $\text{V}_\text{diff} \geq 0$ when $2 \leq M \leq P_u \leq N$.
\end{lemma}

\begin{proof}
Let $f(p) = \frac{4N}{p^2} - \frac{4M+ 6N}{p} - \frac{2p}{M} + 2N + 6$. Then we have,
\begin{equation}
    f'(p) = \frac{4M}{p^2} - \frac{2}{M} + \frac{2N(3p - 4)}{p^3}
\end{equation}
\begin{equation}
\begin{split}
    f''(p) &=4\cdot\frac{6N - (2M + 3N)p}{p^4}\\
    &\leq 4\cdot\frac{6N - 2M^2 - 3MN}{p^4}\\
    &\leq \frac{-32}{p^4}\quad\text{(by setting $M=2$)}\\
    &\leq 0
\end{split}
\end{equation}
Thus, $f'(p)$ is a non-increasing function.
Since $f'(M) = \frac{2}{M} + \frac{2N(3M - 4)}{M^3} \geq 0$, $f'(p)$ has at most one $p^{*}$ such that $f'(p^*) = 0$ and $p^* \in [M, N]$.
If such $p^*$ exists, then $f(p)$ will be increasing on the interval $[M, p^*]$ and be decreasing on the interval $[p^*, N]$.
Thus, the minimum of $f(p)$ would be either $f(M)$ or $f(N)$.
We have,
\begin{equation}
    f(M) = \frac{2N(M - 1)(M - 2)}{M^2} \geq 0
\end{equation}
\begin{equation}
    \begin{split}
        f(N) &= 2N + \frac{4}{N}- \left(\frac{2N}{M} + \frac{4M}{N}\right)\\
        &\geq N - \frac{4}{N}\quad\text{(by setting $M = 2$)}\\
        &\geq 0\quad\text{(by setting $N = 2$)}
    \end{split}
\end{equation}
Thus, we show that $V_\text{diff} = f(P_u) \geq 0$ always holds for $2 \leq M \leq P_u \leq N$.
\end{proof}

\balance

\end{document}